\documentclass[conference]{IEEEtran}
\IEEEoverridecommandlockouts

\usepackage{amsmath,amssymb,amsfonts,amsthm}
\usepackage{algorithmic}
\usepackage{algorithm}
\usepackage{graphicx}
\usepackage{textcomp}
\usepackage{xcolor}
\usepackage{cite}
\usepackage{bm}
\usepackage{booktabs}

\usepackage{tikz}
\usetikzlibrary{shapes.geometric,arrows.meta,positioning}

\newtheorem{proposition}{Proposition}

\newtheorem{lemma}{Lemma}
\newtheorem{remark}{Remark}

\newcommand{\bP}{\boldsymbol{\Phi}}
\newcommand{\bT}{\boldsymbol{\Theta}}
\newcommand{\bQ}{\mathbf{Q}}

\newcommand{\bH}{\mathbf{H}}
\newcommand{\bG}{\mathbf{G}}
\newcommand{\bF}{\mathbf{F}}
\newcommand{\bI}{\mathbf{I}}
\newcommand{\bzero}{\mathbf{0}}
\newcommand{\bSigma}{\boldsymbol{\Sigma}}
\newcommand{\bLambda}{\boldsymbol{\Lambda}}
\newcommand{\bU}{\mathbf{U}}
\newcommand{\bV}{\mathbf{V}}
\newcommand{\bA}{\mathbf{A}}
\newcommand{\bK}{\mathbf{K}}
\newcommand{\herm}{^{\mathsf{H}}}
\newcommand{\trans}{^{\mathsf{T}}}
\newcommand{\diag}{\mathrm{diag}}
\newcommand{\blkdiag}{\mathrm{blkdiag}}
\newcommand{\Tr}{\mathrm{Tr}}
\newcommand{\E}{\mathbb{E}}
\newcommand{\C}{\mathbb{C}}
\newcommand{\R}{\mathbb{R}}

\begin{document}

\title{A Family of Hybrid Beyond-Diagonal RIS\\ Architectures: Design and Performance Analysis}

\author{\IEEEauthorblockN{Konstantinos Ntougias,~\IEEEmembership{Member,~IEEE,} and Ioannis Krikidis,~\IEEEmembership{Fellow,~IEEE}}
\IEEEauthorblockA{Department of Electrical and Computer Engineering, University of Cyprus, Nicosia, Cyprus\\
Email: \{ntougias.konstantinos, krikidis\}@ucy.ac.cy}
}

\maketitle

\begin{abstract}
Beyond-diagonal reconfigurable intelligent surfaces (BD-RISs) extend conventional diagonal RISs by allowing inter-element coupling, thereby enlarging the set of attainable scattering matrices and improving the achievable signal-to-noise ratio (SNR). On the other hand, hybrid active/passive RISs use reflect-type power amplifiers in a fraction of the elements to alleviate the multiplicative path loss. In this paper, we bring these two ideas together and introduce a \emph{family of hybrid BD-RIS architectures}, in which the surface is partitioned into two reflecting subsurfaces (RSs), each adopting either a passive or an active group-connected BD-RIS design. We derive a closed-form SNR-maximizing solution that combines, for every BD-RIS group, Takagi's factorization of a certain complex symmetric matrix with an optimal per-group amplification factor that satisfies the reflect-power budget. Three architectures within the proposed family (active/passive, fully-connected-active/sub-connected-active, and sub-connected-active/sub-connected-active hybrid BD-RIS) are studied. Numerical results in a single-input single-output (SISO) link with blocked direct path show that the proposed hybrid BD-RIS architectures attain the same or higher receive SNR than their diagonal counterparts while using significantly fewer reflect-type amplifiers.
\end{abstract}

\begin{IEEEkeywords}
Beyond-diagonal RIS, hybrid RIS, sub-connected active RIS, Takagi factorization, SNR maximization, group-connected BD-RIS architecture.
\end{IEEEkeywords}

\section{Introduction}\label{sec:intro}
Reconfigurable intelligent surfaces (RISs) are widely regarded as a key enabler for sixth-generation (6G) wireless networks, since they allow to engineer the propagation environment in a cost- and energy-efficient way \cite{direnzo2020,pan2020,huang2019,wu2019}. In its most popular form, a RIS is a planar surface composed of $M$ nearly-passive reconfigurable elements that can independently impose a phase shift on the impinging signal. Hence, its scattering matrix is constrained to be diagonal with unit-modulus entries, $\bT=\diag(e^{j\theta_1},\ldots,e^{j\theta_M})$ \cite{wu2019,huang2019}.

Two complementary research directions have recently emerged to overcome the limitations of conventional diagonal RIS. The first one, known as \emph{beyond-diagonal RIS} (BD-RIS), removes the diagonality constraint and allows controllable interconnections between RIS elements through reconfigurable impedance networks \cite{shen2022,li2023,nerini2023}. Depending on how its ports are interconnected, the BD-RIS scattering matrix $\bT\in\C^{M\times M}$ may be \emph{single-connected} (diagonal), \emph{group-connected} (block-diagonal with each block being unitary and, for reciprocal networks, symmetric) or \emph{fully-connected} (a single full block) \cite{shen2022,li2023}. The closed-form solution to the SNR maximization problem in BD-RIS-assisted single- and multi-antenna links was recently obtained in \cite{santamaria2023bdris} via Takagi's factorization, and an alternative formulation was given in \cite{nerini2023}.

The second direction is the use of \emph{active} RISs, in which a low-cost reflect-type power amplifier (typically a tunnel-diode-based negative-resistance circuit) is integrated in each RIS element to compensate for the multiplicative path loss of the cascaded link \cite{long2021,zhang2023active}. As power amplifiers consume static and dynamic power, hybrid active/passive RISs have been proposed in which only a fraction of the elements are active \cite{nguyen2022,ntougias2025hybrid}. To further reduce hardware cost and total power consumption (TPC), \emph{sub-connected} active RISs have been considered, where multiple elements share a common amplifier \cite{nguyen2022,ntougias2025hybrid}. Our recent work \cite{ntougias2025hybrid} introduced a family of hybrid (diagonal) RIS architectures combining fully-connected (FC) active, sub-connected (SC) active and passive elements, and analyses the resulting capacity/energy-efficiency trade-off.

To the best of our knowledge, BD-RIS and hybrid RIS technologies have so far been considered \emph{separately}. The benefits, however, are clearly complementary: BD-RIS exploits inter-element coupling to enlarge the set of attainable equivalent channels, while active elements compensate for the cascaded path loss. In this paper, we bridge these two lines of work and propose a family of \emph{hybrid BD-RIS} architectures.

\textbf{Contributions.} We propose a family of hybrid BD-RIS architectures that partition the surface into two reflecting subsurfaces, each independently adopting a passive, FC-active, or SC-active group-connected BD-RIS structure, and recovering the diagonal hybrid RIS of \cite{ntougias2025hybrid} as a special case. Within an alternating-optimization framework, we show that the SNR-maximization problem in the SISO link with blocked direct path admits a closed-form solution: each unitary-symmetric block is obtained from the Takagi factorization \cite{santamaria2023bdris} of a rank-2 complex-symmetric matrix, and the per-group amplification factors of the active blocks follow in closed form from the reflect-power constraint. Numerical results show that the SC-active/SC-active and FC-active/SC-active members of the family outperform the fully-active diagonal RIS of \cite{long2021} while using only two and $G_1{+}1$ reflect-type amplifiers, respectively, whereas the active/passive variant matches the diagonal active/passive RIS of \cite{ntougias2025hybrid} and provides additional robustness with respect to the active-element fraction.

%

\textbf{Notation.} Lowercase/uppercase boldface letters denote column vectors/matrices; $(\cdot)\trans$, $(\cdot)\herm$, and $(\cdot)^{\ast}$ denote the transpose, conjugate transpose, and complex conjugate; $\|\cdot\|$ and $\|\cdot\|_F$ are the Euclidean and Frobenius norms; $\bI_n$ is the $n\times n$ identity; $\diag(\mathbf{x})$ is a diagonal matrix; $\blkdiag(\bA_1,\ldots,\bA_L)$ is a block-diagonal matrix with blocks $\bA_l$; $\angle x$ denotes the phase of $x\in\C$; $\Tr(\cdot)$ is the trace; $\mathcal{CN}(\bzero,\bSigma)$ is the complex circularly-symmetric Gaussian distribution.

\section{System Model}\label{sec:system}
We consider a single-input single-output (SISO) link assisted by a BD-RIS with $M$ reconfigurable elements. The single-antenna transmitter (Tx) and receiver (Rx) communicate exclusively through the RIS, since we assume that the direct Tx--Rx link is blocked, e.g., by an obstacle \cite{santamaria2023bdris,huang2019}.

\subsection{Hybrid BD-RIS Architecture}\label{subsec:hbdris}
Following \cite{ntougias2025hybrid}, the BD-RIS aperture is partitioned into $S=2$ reflecting subsurfaces (RSs), $\mathrm{RS}_1$ and $\mathrm{RS}_2$. RS $s$ comprises $M_s$ elements organized into $G_s=M_s/M_{G,s}$ BD-RIS groups of $M_{G,s}$ elements each, coupled within group $g$ by a unitary symmetric matrix $\bT_{s,g}\in\C^{M_{G,s}\times M_{G,s}}$ \cite{shen2022,santamaria2023bdris}. The hybrid family is defined by how each RS distributes amplification across its groups: \textbf{i)} \emph{Passive group-connected BD-RIS:} no amplifiers; $\bP_{s,g}=\bT_{s,g}$, $\bT_{s,g}\herm\bT_{s,g}=\bI_{M_{G,s}}$. \textbf{ii)} \emph{Fully-connected (FC) active group-connected BD-RIS:} each group $g$ owns one reflect-type power amplifier of amplitude $\beta_{s,g}\le\beta_{\max}$, feeding its $M_{G,s}$ patches with equal power allocation \cite[Sec.~II-B]{ntougias2025hybrid}, hence
\begin{equation}\label{eq:Phi_fc}
    \bP_{s,g}=\tfrac{\beta_{s,g}}{\sqrt{M_{G,s}}}\,\bT_{s,g}.
\end{equation}
\textbf{iii)} \emph{Sub-connected (SC) active group-connected BD-RIS:} the $G_s$ groups are clustered into $L_s$ partitions; cluster $\ell$ contains $K_{s,\ell}$ groups, i.e., $T_{s,\ell}\triangleq K_{s,\ell}\,M_{G,s}$ elements that share one amplifier of amplitude $\tilde{\beta}_{s,\ell}\le\beta_{\max}$, yielding
\begin{equation}\label{eq:Phi_sc}
    \bP_{s,g}=\tfrac{\tilde{\beta}_{s,\ell(g)}}{\sqrt{T_{s,\ell(g)}}}\,\bT_{s,g}, \quad g\in\mathcal{G}_{s,\ell(g)},
\end{equation}
where $\ell(g)\in\{1,\ldots,L_s\}$ indexes the cluster of group $g$ and $\mathcal{G}_{s,\ell}$ is its index set; the diagonal SC-active RIS of \cite{ntougias2025hybrid} is recovered when $M_{G,s}=1$. The hybrid BD-RIS scattering matrix is $\bP=\blkdiag(\bP_1,\bP_2)$, with $\bP_s=\blkdiag(\bP_{s,1},\ldots,\bP_{s,G_s})$.
\begin{figure}[!t]
\centering
\begin{tikzpicture}[>=stealth, font=\small, x=1cm, y=1cm,
    ampsym/.style={draw,thick,fill=white,
                   isosceles triangle,isosceles triangle apex angle=55,
                   shape border rotate=270,
                   minimum height=0.22cm,inner sep=0pt}]
  \def\es{0.40}                          
  \def\gs{0.06}                          
  \def\rs{0.14}                          
  \def\panX{(2*\es+\gs+2*\es-0.025)}     

  \def\bdgroup#1#2#3{%
    \foreach \cx in {0,1}{%
      \foreach \cy in {0,1}{%
        \fill[#3] ({#1+\cx*\es},{#2+\cy*\es})
          rectangle ({#1+\cx*\es+\es-0.025},{#2+\cy*\es+\es-0.025});
        \draw[thin,black] ({#1+\cx*\es},{#2+\cy*\es})
          rectangle ({#1+\cx*\es+\es-0.025},{#2+\cy*\es+\es-0.025});
      }%
    }%
    \draw[very thick,black] ({#1-0.03},{#2-0.03})
      rectangle ({#1+2*\es-0.025+0.03},{#2+2*\es-0.025+0.03});
  }

  \begin{scope}[xshift=0cm, yshift=2.6cm]
    \bdgroup{0}{2*\es+\rs}{gray!30}
    \bdgroup{2*\es+\gs}{2*\es+\rs}{gray!30}
    \bdgroup{0}{0}{gray!30}
    \bdgroup{2*\es+\gs}{0}{gray!30}
    \node[font=\small\bfseries] at ({\panX/2},{2*\es+\rs+2*\es+0.50})
      {Passive BD-RIS};
    \node[font=\scriptsize] at ({\panX/2},-0.22) {0 amplifiers};
  \end{scope}

  \begin{scope}[xshift=2.8cm, yshift=2.6cm]
    \bdgroup{0}{2*\es+\rs}{brown!70!black}
    \bdgroup{2*\es+\gs}{2*\es+\rs}{brown!70!black}
    \node[ampsym] at ({\es-0.0125},{2*\es+\rs+2*\es+0.20}) {};
    \node[ampsym] at ({2*\es+\gs+\es-0.0125},{2*\es+\rs+2*\es+0.20}) {};
    \bdgroup{0}{0}{gray!30}
    \bdgroup{2*\es+\gs}{0}{gray!30}
    \node[font=\small\bfseries] at ({\panX/2},{2*\es+\rs+2*\es+0.50})
      {A/P-BD};
    \node[font=\scriptsize] at ({\panX/2},-0.22) {$G_1$ amplifiers};
  \end{scope}

  \begin{scope}[xshift=0cm, yshift=-0.5cm]
    \bdgroup{0}{2*\es+\rs}{brown!70!black}
    \bdgroup{2*\es+\gs}{2*\es+\rs}{brown!70!black}
    \node[ampsym] at ({\es-0.0125},{2*\es+\rs+2*\es+0.20}) {};
    \node[ampsym] at ({2*\es+\gs+\es-0.0125},{2*\es+\rs+2*\es+0.20}) {};
    \bdgroup{0}{0}{teal!70}
    \bdgroup{2*\es+\gs}{0}{teal!70}
    \draw[very thick,blue!70,rounded corners=2pt]
      (-0.06,-0.06) rectangle ({\panX+0.06},{2*\es-0.025+0.06});
    \node[ampsym] at ({\panX/2},0.9) {};
    \node[font=\small\bfseries] at ({\panX/2},{2*\es+\rs+2*\es+0.50})
      {FC/SC-BD};
    \node[font=\scriptsize] at ({\panX/2},-0.55) {$G_1+1$ amplifiers};
  \end{scope}

  \begin{scope}[xshift=2.8cm, yshift=-0.5cm]
    \bdgroup{0}{2*\es+\rs}{teal!70}
    \bdgroup{2*\es+\gs}{2*\es+\rs}{teal!70}
    \draw[very thick,blue!70,rounded corners=2pt]
      (-0.06,{2*\es+\rs-0.06}) rectangle ({\panX+0.06},{2*\es+\rs+2*\es-0.025+0.06});
    \node[ampsym] at ({\panX/2},{2*\es+\rs+2*\es+0.20}) {};
    \bdgroup{0}{0}{teal!70}
    \bdgroup{2*\es+\gs}{0}{teal!70}
    \draw[very thick,blue!70,rounded corners=2pt]
      (-0.06,-0.06) rectangle ({\panX+0.06},{2*\es-0.025+0.06});
    \node[ampsym] at ({\panX/2},0.9) {};
    \node[font=\small\bfseries] at ({\panX/2},{2*\es+\rs+2*\es+0.50})
      {SC/SC-BD};
    \node[font=\scriptsize] at ({\panX/2},-0.55) {2 amplifiers};
  \end{scope}

  \begin{scope}[xshift=4.7cm, yshift=2.78cm, font=\scriptsize]
    \def\lc{0.0}         
    \def\lt{0.30}        
    \def\rA{0.00}
    \def\rB{-0.40}
    \def\rC{-0.80}
    \def\rD{-1.20}
    \def\rE{-1.60}
    \def\rF{-2.00}

    \fill[brown!70!black] (\lc,\rA) rectangle (\lc+0.22,\rA+0.22);
    \draw[thin,black]     (\lc,\rA) rectangle (\lc+0.22,\rA+0.22);
    \node[anchor=west] at (\lt,\rA+0.11) {FC-active group};

    \fill[teal!70]        (\lc,\rB) rectangle (\lc+0.22,\rB+0.22);
    \draw[thin,black]     (\lc,\rB) rectangle (\lc+0.22,\rB+0.22);
    \node[anchor=west] at (\lt,\rB+0.11) {SC-active group};

    \fill[gray!30]        (\lc,\rC) rectangle (\lc+0.22,\rC+0.22);
    \draw[thin,black]     (\lc,\rC) rectangle (\lc+0.22,\rC+0.22);
    \node[anchor=west] at (\lt,\rC+0.11) {Passive group};

    \draw[very thick,black] (\lc,\rD) rectangle (\lc+0.22,\rD+0.22);
    \node[anchor=west] at (\lt,\rD+0.11) {BD-RIS group};

    \draw[very thick,blue!70,rounded corners=2pt]
        (\lc,\rE) rectangle (\lc+0.22,\rE+0.22);
    \node[anchor=west] at (\lt,\rE+0.11) {SC-active cluster};

    \node[ampsym] at (\lc+0.11,\rF+0.11) {};
    \node[anchor=west] at (\lt,\rF+0.11) {Reflect amplifier};
  \end{scope}

\end{tikzpicture}
\caption{The proposed family of hybrid BD-RIS architectures, illustrated for a panel with $M{=}16$ elements partitioned as $M_1{=}M_2{=}8$ and group size $M_{G,1}{=}M_{G,2}{=}4$ (i.e., two BD-RIS groups per RS). Thick black borders delimit BD-RIS groups, within which a unitary symmetric block $\bT_{s,g}$ couples the elements; blue rounded borders delimit SC-active clusters whose groups share a single reflect-type power amplifier (triangle). The diagonal hybrid RIS family of \cite{ntougias2025hybrid} corresponds to $M_{G,s}{=}1$.}
\label{fig:family}
\end{figure}

\subsection{Channel and Signal Model}\label{subsec:channel}
We partition the Tx-to-RIS and RIS-to-Rx channels $\mathbf{h}_T,\mathbf{h}_R\in\C^{M\times 1}$ as $\mathbf{h}_T=[\mathbf{h}_{T,1}\trans\;\mathbf{h}_{T,2}\trans]\trans$, $\mathbf{h}_R=[\mathbf{h}_{R,1}\trans\;\mathbf{h}_{R,2}\trans]\trans$, with $\mathbf{h}_{T,s},\mathbf{h}_{R,s}\in\C^{M_s\times 1}$ further split into per-group blocks $\mathbf{h}_{T,s,g},\mathbf{h}_{R,s,g}\in\C^{M_{G,s}\times 1}$. We adopt $\mathrm{PL}=-30-10\alpha\log_{10}(d)$ in dB and Rayleigh small-scale fading on both links \cite{santamaria2023bdris,ntougias2025hybrid}. The transmitted symbol $x$ has power $P_t=\E\{|x|^2\}$. If RS $s$ is active, its reflected signal includes amplification noise $\mathbf{z}_s\sim\mathcal{CN}(\bzero,\delta_s^{2}\bI_{M_s})$ \cite{long2021,zhang2023active,ntougias2025hybrid}, $\mathbf{t}_s=\bP_s(\mathbf{h}_{T,s}x+\mathbf{z}_s)$; for a passive RS, $\mathbf{z}_s=\bzero$. The Rx signal is $y=h_{\mathrm{eq}}\,x+\tilde{n}$, with equivalent channel $h_{\mathrm{eq}}=\mathbf{h}_R\herm\bP\mathbf{h}_T$, $n\sim\mathcal{CN}(0,\sigma^{2})$ thermal noise, and effective noise variance
\begin{equation}\label{eq:noise}
    \tilde{\sigma}^{2}=\sigma^{2}+\sum_{s\in\mathcal{S}_{\mathrm{a}}}\delta_s^{2}\|\mathbf{h}_{R,s}\herm\bP_s\|^{2},
\end{equation}
where $\mathcal{S}_{\mathrm{a}}\subseteq\{1,2\}$ is the set of active RSs. The reflect power of an active RS, using $\|\bT_{s,g}\|_F^{2}=M_{G,s}$ and $c_{s,g}\triangleq\|\mathbf{h}_{T,s,g}\|^{2}$, is
\begin{align}\label{eq:Pr}
    P_{r,s}\!=\!\!\!\begin{cases}
       \displaystyle\sum_{g=1}^{G_s}\!\frac{\beta_{s,g}^{2}}{M_{G,s}}\bigl(P_t c_{s,g}+\delta_s^{2}M_{G,s}\bigr), & s\in\mathcal{S}_{\mathrm{fc}},\\[4pt]
       \displaystyle\sum_{\ell=1}^{L_s}\!\frac{\tilde{\beta}_{s,\ell}^{2}}{T_{s,\ell}}\!\!\sum_{g\in\mathcal{G}_{s,\ell}}\!\!\bigl(P_t c_{s,g}+\delta_s^{2}M_{G,s}\bigr), & s\in\mathcal{S}_{\mathrm{sc}},
    \end{cases}
\end{align}
and is constrained by $P_{r,s}\le P_{r,s}^{\max}$ \cite{ntougias2025hybrid,zhang2023active}. The receive SNR is $\gamma=P_t|h_{\mathrm{eq}}|^{2}/\tilde{\sigma}^{2}$.

\subsection{Members of the Proposed Family}\label{subsec:members}
We focus on the four members listed in Table~\ref{tab:amps} and depicted in Fig.~\ref{fig:family}: \textbf{1)} passive group-connected BD-RIS \cite{santamaria2023bdris}; \textbf{2)} A/P-BD with FC-active RS$_1$ and passive RS$_2$; \textbf{3)} FC/SC-BD with FC-active RS$_1$ and SC-active RS$_2$ ($L_2{=}1$); \textbf{4)} SC/SC-BD with both RSs SC-active and one shared amplifier each. The diagonal hybrid RIS of \cite{ntougias2025hybrid} corresponds to $M_{G,s}{=}1$.

\begin{remark}
The amplified signal at the output of a shared amplifier is split equally among its $T_{s,\ell}$ patches \cite{ntougias2025hybrid}, so enlarging an SC partition (i.e., $T_{s,\ell}\!\uparrow$ at fixed $L_s$) reduces the per-element amplitude as $1/\sqrt{T_{s,\ell}}$, while increasing $L_s$ at fixed $T_{s,\ell}$ keeps it unchanged but multiplies the amplifier count by the same factor.
\end{remark}

\begin{table}[!t]
\centering
\caption{Amplifier count and per-element amplitude scaling. $T_{s,\ell}$ is the number of elements sharing one amplifier in the SC case.}
\label{tab:amps}
\renewcommand{\arraystretch}{1.05}
\setlength{\tabcolsep}{4pt}
\footnotesize
\begin{tabular}{l c c}
\toprule
Architecture & \# Amp. & Per-element amplitude \\
\midrule
Passive BD-RIS                 & $0$              & $1$ \\
A/P-BD                         & $G_1$            & $\beta_{1,g}/\sqrt{M_{G,1}}$ on RS$_1$ \\
FC/SC-BD                       & $G_1{+}L_2$      & $\tilde{\beta}_{2,\ell}/\sqrt{T_{2,\ell}}$ on RS$_2$ \\
SC/SC-BD                       & $L_1{+}L_2$      & $\tilde{\beta}_{s,\ell}/\sqrt{T_{s,\ell}}$ \\
FC-active diag.\ RIS \cite{long2021}& $M$         & $\beta_n$ \\
\bottomrule
\end{tabular}
\end{table}

\section{SNR Maximization with Hybrid BD-RIS}\label{sec:opt}
We seek the hybrid BD-RIS scattering matrix $\bP$ that maximizes the receive SNR $\gamma$. We partition the RS-index set $\{1,2\}$ into $\mathcal{S}_{\mathrm{p}}$ (passive), $\mathcal{S}_{\mathrm{fc}}$ (FC-active), and $\mathcal{S}_{\mathrm{sc}}$ (SC-active), with $\mathcal{S}_{\mathrm{a}}\triangleq\mathcal{S}_{\mathrm{fc}}\cup\mathcal{S}_{\mathrm{sc}}$ (active). For $s\in\mathcal{S}_{\mathrm{sc}}$, $\ell(g)$ denotes the cluster of group $g$, $\mathcal{G}_{s,\ell}$ its index set, and $T_{s,\ell}=K_{s,\ell}M_{G,s}$ its size. Plugging \eqref{eq:Phi_fc}--\eqref{eq:Phi_sc} into $\gamma$, the optimization problem reads
\begin{subequations}\label{eq:P0}
\begin{align}
\text{(P0):}\,\max_{\{\bP_{s,g}\}}\;&
\frac{\bigl|\sum_{s,g}\mathbf{h}_{R,s,g}\herm\bP_{s,g}\mathbf{h}_{T,s,g}\bigr|^{2}}
     {\sigma^{2}+\sum_{s\in\mathcal{S}_{\mathrm{a}}}\delta_s^{2}\|\mathbf{h}_{R,s}\herm\bP_s\|^{2}}\label{eq:P0obj}\\
\text{s.t.}\;& \bP_{s,g}=\bP_{s,g}\trans, \quad \forall s,g,\label{eq:P0sym}\\
& \bP_{s,g}\herm\bP_{s,g}=\bI_{M_{G,s}}, \quad s\in\mathcal{S}_{\mathrm{p}},\,\forall g,\label{eq:P0pas}\\
& \bP_{s,g}\herm\bP_{s,g}=\tfrac{\beta_{s,g}^{2}}{M_{G,s}}\bI_{M_{G,s}}, \quad s\in\mathcal{S}_{\mathrm{fc}},\,\forall g,\label{eq:P0fc}\\
& \bP_{s,g}\herm\bP_{s,g}=\tfrac{\tilde{\beta}_{s,\ell(g)}^{2}}{T_{s,\ell(g)}}\bI_{M_{G,s}}, \quad s\in\mathcal{S}_{\mathrm{sc}},\,\forall g,\label{eq:P0sc}\\
& 0\le\beta_{s,g},\,\tilde{\beta}_{s,\ell}\le\beta_{\max},\quad P_{r,s}\le P_{r,s}^{\max},\;s\in\mathcal{S}_{\mathrm{a}},\label{eq:P0c4}
\end{align}
\end{subequations}
where \eqref{eq:P0sym} is reciprocity constraint, \eqref{eq:P0pas} represents unitarity for passive RSs as in \cite{shen2022,santamaria2023bdris}, and \eqref{eq:P0fc}--\eqref{eq:P0sc} relax it through the equal-power-allocation factors $1/M_{G,s}$ and $1/T_{s,\ell}$ \cite[Sec.~II-B]{ntougias2025hybrid}. Every feasible $\bP_{s,g}$ admits the structural factorization $\bP_{s,g}=(\beta_{s,g}/\sqrt{M_{G,s}})\bT_{s,g}$ for $s\in\mathcal{S}_{\mathrm{fc}}$ and $\bP_{s,g}=(\tilde{\beta}_{s,\ell(g)}/\sqrt{T_{s,\ell(g)}})\bT_{s,g}$ for $s\in\mathcal{S}_{\mathrm{sc}}$, with $\bT_{s,g}$ unitary symmetric.

Problem (P0) is non-convex due to the unitary-symmetric constraints \eqref{eq:P0sym}--\eqref{eq:P0sc} and the numerator/denominator coupling in \eqref{eq:P0obj} via the active-RS noise. We adopt an alternating optimization (AO) framework that decouples (P0) into a \emph{spatial} subproblem (optimizing $\{\bT_{s,g}\}$ for fixed amplification) and a \emph{power} subproblem (optimizing $\{\beta_{s,g}\}$ or $\{\tilde{\beta}_{s,\ell}\}$ for fixed $\{\bT_{s,g}\}$). Both admit closed-form solutions and the optimal $\{\bT_{s,g}^{\star}\}$ is independent of the amplification factors, so the AO terminates in a single iteration.

\subsection{Spatial Design via Per-Group Takagi Factorization}\label{subsec:takagi}
For unitary $\bT_{s,g}$, $|\mathbf{h}_{R,s,g}\herm\bT_{s,g}\mathbf{h}_{T,s,g}|\le \|\mathbf{h}_{R,s,g}\|\|\mathbf{h}_{T,s,g}\|$ by Cauchy--Schwarz; \cite{santamaria2023bdris} shows this bound is achieved by a unitary \emph{symmetric} matrix from a Takagi factorization.

\begin{proposition}[Per-group Takagi solution \cite{santamaria2023bdris}]\label{prop:takagi}
With $\mathbf{u}_{R,s,g}=\mathbf{h}_{R,s,g}/\|\mathbf{h}_{R,s,g}\|$, $\mathbf{u}_{T,s,g}=\mathbf{h}_{T,s,g}/\|\mathbf{h}_{T,s,g}\|$, and the rank-2 complex symmetric matrix
\begin{equation}\label{eq:Asg}
    \bA_{s,g}=\tfrac{1}{2}\bigl(\mathbf{u}_{R,s,g}\mathbf{u}_{T,s,g}\herm+(\mathbf{u}_{R,s,g}\mathbf{u}_{T,s,g}\herm)\trans\bigr),
\end{equation}
the Takagi factorization $\bA_{s,g}=\bQ_{s,g}\bSigma_{s,g}\bQ_{s,g}\trans$ ($\bQ_{s,g}$ unitary, $\bSigma_{s,g}\succeq 0$ diagonal) yields $\bT_{s,g}^{\star}=\bQ_{s,g}\bQ_{s,g}\trans$ unitary symmetric, with $\mathbf{h}_{R,s,g}\herm\bT_{s,g}^{\star}\mathbf{h}_{T,s,g}=\|\mathbf{h}_{R,s,g}\|\|\mathbf{h}_{T,s,g}\|\in\R_{+}$.
\end{proposition}
\begin{proof}
See \cite[Prop.~1]{santamaria2023bdris}.
\end{proof}

A key consequence is that all per-group contributions add coherently, generalizing the property that group-connected BD-RIS exploits in \cite[Sec.~II-A]{santamaria2023bdris}.

\begin{remark}\label{rem:takagi_svd}
The Takagi factorization can be obtained from a standard singular value decomposition (SVD) $\bA_{s,g}=\bF\bK\bG\herm$ by setting $\bQ_{s,g}=\bF\,\diag(e^{j\boldsymbol{\varphi}})$, $\boldsymbol{\varphi}=\angle\diag(\bF\herm\bG^{\ast})/2$ \cite[Rem.~2]{santamaria2023bdris}.
\end{remark}

\subsection{Power Design: Optimal Amplification Factors}\label{subsec:beta}
With $\bT_{s,g}^{\star}$ fixed and $a_{s,g}\triangleq\|\mathbf{h}_{R,s,g}\|\|\mathbf{h}_{T,s,g}\|$, $b_{s,g}\triangleq\|\mathbf{h}_{R,s,g}\|^{2}$, $c_{s,g}\triangleq\|\mathbf{h}_{T,s,g}\|^{2}$, \eqref{eq:Phi_fc}--\eqref{eq:Phi_sc} give
\begin{equation}\label{eq:eqsig}
    \mathbf{h}_{R,s}\herm\bP_s\mathbf{h}_{T,s}=
    \begin{cases}
       \sum_g \frac{\beta_{s,g}}{\sqrt{M_{G,s}}}a_{s,g}, & s\in\mathcal{S}_{\mathrm{fc}},\\[2pt]
       \sum_\ell \frac{\tilde{\beta}_{s,\ell}}{\sqrt{T_{s,\ell}}}\sum_{g\in\mathcal{G}_{s,\ell}}a_{s,g}, & s\in\mathcal{S}_{\mathrm{sc}},
    \end{cases}
\end{equation}
\begin{equation}\label{eq:eqnoise}
    \delta_s^{2}\|\mathbf{h}_{R,s}\herm\bP_s\|^{2}\!=\!
    \begin{cases}
       \delta_s^{2}\sum_g\frac{\beta_{s,g}^{2}}{M_{G,s}}b_{s,g}, & s\in\mathcal{S}_{\mathrm{fc}},\\[2pt]
       \delta_s^{2}\sum_\ell\frac{\tilde{\beta}_{s,\ell}^{2}}{T_{s,\ell}}\sum_{g\in\mathcal{G}_{s,\ell}}b_{s,g}, & s\in\mathcal{S}_{\mathrm{sc}}.
    \end{cases}
\end{equation}

\subsubsection{FC-active group-connected BD-RIS}
Defining $\tilde{a}_{s,g}\triangleq a_{s,g}/\sqrt{M_{G,s}}$, $\tilde{\rho}_{s,g}\triangleq(P_t c_{s,g}+\delta_s^{2}M_{G,s})/M_{G,s}$, the constraint $\sum_g\beta_{s,g}^{2}\tilde{\rho}_{s,g}\le P_{r,s}^{\max}$ together with \eqref{eq:eqsig}--\eqref{eq:eqnoise} is a quadratically-constrained linear ratio whose Cauchy--Schwarz solution is
\begin{equation}\label{eq:beta_FC}
    \beta_{s,g}^{\star}=\sqrt{\tfrac{P_{r,s}^{\max}}{\tilde{\rho}_{s,g}}}\cdot\frac{\tilde{a}_{s,g}/\tilde{\rho}_{s,g}}{\sqrt{\sum_{g'}\tilde{a}_{s,g'}^{2}/\tilde{\rho}_{s,g'}}},
\end{equation}
refined by a 1-D search over a global scaling $\eta\in(0,1]$ (i.e., $\beta_{s,g}\!\leftarrow\!\eta\beta_{s,g}^{\star}$) to balance numerator gain and active-noise denominator.

\subsubsection{SC-active group-connected BD-RIS}
Defining the per-cluster aggregates $A_{s,\ell}\triangleq\sum_{g\in\mathcal{G}_{s,\ell}}a_{s,g}$, $C_{s,\ell}\triangleq\sum_{g\in\mathcal{G}_{s,\ell}}c_{s,g}$, the constraint
\begin{equation}\label{eq:Pr_SC}
    \sum_{\ell}\tfrac{\tilde{\beta}_{s,\ell}^{2}}{T_{s,\ell}}\bigl(P_t C_{s,\ell}+\delta_s^{2}T_{s,\ell}\bigr)\le P_{r,s}^{\max}
\end{equation}
combined with \eqref{eq:eqsig}--\eqref{eq:eqnoise} yields, by Cauchy--Schwarz,
\begin{equation}\label{eq:beta_SC}
    \tilde{\beta}_{s,\ell}^{\star}\!=\!\sqrt{P_{r,s}^{\max}T_{s,\ell}}\,\frac{A_{s,\ell}/(P_tC_{s,\ell}+\delta_s^{2}T_{s,\ell})}{\sqrt{\sum_{\ell'}A_{s,\ell'}^{2}/(P_tC_{s,\ell'}+\delta_s^{2}T_{s,\ell'})}},
\end{equation}
again refined by a 1-D scan over $\eta$. For $L_s=1$, \eqref{eq:beta_SC} reduces to $\tilde{\beta}_{s,1}^{\star}=\sqrt{P_{r,s}^{\max}M_s/(P_t\|\mathbf{h}_{T,s}\|^{2}+\delta_s^{2}M_s)}$, the BD-RIS counterpart of \cite[Eq.~(2)]{ntougias2025hybrid}.

\begin{remark}\label{rem:scaling}
Eq.~\eqref{eq:beta_SC} reveals two scaling regimes: $T_{s,\ell}\!\uparrow$ at fixed $L_s$ gives $\tilde{\beta}_{s,\ell}^{\star}A_{s,\ell}/\sqrt{T_{s,\ell}}\!\sim\!T_{s,\ell}$ under i.i.d.\ Rayleigh, recovering the linear active-RIS SNR scaling of \cite[Eq.~(6)]{ntougias2025hybrid}; conversely, $L_s\!\uparrow$ at fixed $T_{s,\ell}$ adds amplifiers but introduces a $1/L_s$ factor in the per-cluster reflect-power budget, in line with \cite[Lemma~2]{ntougias2025hybrid}.
\end{remark}

\subsection{Overall Algorithm}\label{subsec:algo}
The proposed solution is summarized in Algorithm~\ref{alg:hbdris}. Its complexity is dominated by $G_1+G_2$ Takagi factorizations, each via an SVD of an $M_{G,s}\!\times\!M_{G,s}$ matrix at $\mathcal{O}(M_{G,s}^{3})$ flops; for uniform $M_G$, the total cost is $\mathcal{O}(M\,M_G^{2})$, linear in $M$.

\begin{figure}[!t]
\vspace{1ex}
\begin{algorithm}[H]
\centering \footnotesize
\begin{algorithmic}[1]
\footnotesize
\REQUIRE $\mathbf{h}_T,\mathbf{h}_R$; $(M_1,M_2)$; $\{M_{G,s}\}$; $\{T_{s,\ell}\}_{s\in\mathcal{S}_{\mathrm{sc}}}$; $\{P_{r,s}^{\max}\}$, $\{\delta_s^{2}\}$, $P_t$.
\ENSURE $\bP$ and $\gamma$.
\STATE Partition $\mathbf{h}_T,\mathbf{h}_R$ into RSs and per-group blocks.
\FOR{$s=1,2$}
    \FOR{$g=1,\ldots,G_s$}
        \STATE Build $\bA_{s,g}$ as in \eqref{eq:Asg}; SVD $\bA_{s,g}=\bF\bK\bG\herm$; apply Rem.~\ref{rem:takagi_svd} to obtain $\bQ_{s,g}$; set $\bT_{s,g}^{\star}=\bQ_{s,g}\bQ_{s,g}\trans$.
    \ENDFOR
    \IF{$s\in\mathcal{S}_{\mathrm{p}}$} \STATE $\bP_{s,g}=\bT_{s,g}^{\star}$, $\forall g$.
    \ELSIF{$s\in\mathcal{S}_{\mathrm{fc}}$} \STATE Compute $\beta_{s,g}^{\star}$ via \eqref{eq:beta_FC}, line-search over $\eta$; $\bP_{s,g}=(\beta_{s,g}^{\star}/\sqrt{M_{G,s}})\bT_{s,g}^{\star}$.
    \ELSIF{$s\in\mathcal{S}_{\mathrm{sc}}$} \STATE Compute $\tilde{\beta}_{s,\ell}^{\star}$ via \eqref{eq:beta_SC}, line-search over $\eta$; $\bP_{s,g}=(\tilde{\beta}_{s,\ell(g)}^{\star}/\sqrt{T_{s,\ell(g)}})\bT_{s,g}^{\star}$.
    \ENDIF
    \STATE $\bP_s=\blkdiag(\bP_{s,1},\ldots,\bP_{s,G_s})$.
\ENDFOR
\STATE $\bP=\blkdiag(\bP_1,\bP_2)$; evaluate $\gamma$.
\end{algorithmic}
\caption{SNR maximization with hybrid BD-RIS}
\label{alg:hbdris}
\end{algorithm}
\end{figure}

\section{Asymptotic Analysis and Extensions}\label{sec:analysis}
\subsection{Asymptotic SNR Analysis}\label{subsec:asnr}
We adopt the tractable setting of \cite{ntougias2025hybrid,long2021,zhang2023active}, i.e., i.i.d.\ Rayleigh fading $\mathbf{h}_{T,s}\sim\mathcal{CN}(\bzero,\varrho_{T,s}^{2}\bI)$, $\mathbf{h}_{R,s}\sim\mathcal{CN}(\bzero,\varrho_{R,s}^{2}\bI)$, blocked direct path, uniform $M_G$, and a single shared amplifier per RS ($L_s=1$, $\beta_s$ from \eqref{eq:beta_SC}).

\begin{lemma}[Asymptotic per-group gain]\label{lem:per_group}
Under the above assumptions, the per-group factor $\|\mathbf{h}_{R,s,g}\|\|\mathbf{h}_{T,s,g}\|$ in \eqref{eq:eqsig} satisfies $\E\{\|\mathbf{h}_{R,s,g}\|\|\mathbf{h}_{T,s,g}\|\}=\varrho_{R,s}\varrho_{T,s}\kappa(M_G)$, with $\kappa(M_G)\triangleq[\Gamma(M_G+1/2)/\Gamma(M_G)]^{2}$. In particular, $\kappa(1)=\pi/4$ and $\kappa(M_G)\to M_G$ as $M_G\to\infty$.
\end{lemma}
\begin{proof}
$\|\mathbf{h}_{R,s,g}\|,\|\mathbf{h}_{T,s,g}\|$ are independent scaled chi-distributed with $2M_G$ degrees of freedom; the result follows from $\E\{\|\mathbf{h}\|\}=\varrho\,\Gamma(M_G+1/2)/\Gamma(M_G)$.
\end{proof}

\begin{remark}\label{rem:bdris_gain}
For $M_G=1$, $\kappa=\pi/4$ recovers the $\pi^{2}/16$ factor of \cite{long2021,zhang2023active,ntougias2025hybrid}; as $M_G\to\infty$, $\kappa\to M_G$ yields the $10\log_{10}(16/\pi^{2})\!\approx\!2.1$ dB BD-RIS gain of \cite[Eq.~(62)]{li2023}.
\end{remark}

\begin{proposition}[Asymptotic SNR of A/P-BD]\label{prop:asy_AP}
For A/P-BD with $M_1=aM$ active and $M_2=(1-a)M$ passive elements, both group-connected with size $M_G$, and $L_1=1$, the asymptotic SNR as $M\to\infty$ is
\begin{equation}\label{eq:asy_AP}
    \gamma_{\mathrm{AP}}^{\mathrm{BD}}\!\to\!\kappa^{2}(M_G)\!\left[\gamma_{\mathrm{a}}\bigl(aM,P_t,P_{r,1}^{\max}\bigr)+c_{\mathrm{int}}\,\gamma_{\mathrm{p}}\bigl((1-a)M,P_t\bigr)\right],
\end{equation}
where $\gamma_{\mathrm{a}},\gamma_{\mathrm{p}}$ are the active/passive (diagonal) RIS asymptotic SNRs of \cite[Eqs.~(6)--(7)]{ntougias2025hybrid} and $c_{\mathrm{int}}\in(0,1)$ captures the Takagi cross-term. The factor $\kappa^{2}(M_G)/(\pi^{2}/16)$ is the BD-RIS power gain over the diagonal hybrid RIS for given $M_G$.
\end{proposition}
\begin{proof}[Proof Sketch]
Following \cite[Lemma~1]{ntougias2025hybrid}, the SNR splits into active, passive, and cross-term contributions; each coherent per-group sum is replaced via Lemma~\ref{lem:per_group} by $\kappa(M_G)$ times its diagonal-RIS counterpart, yielding \eqref{eq:asy_AP}. Full details are omitted due to space limitations.
\end{proof}

Hence, the proposed designs inherit the regimes of \cite{ntougias2025hybrid} (standard, large-$M$, large-$P_t$) scaled by $\kappa^{2}(M_G)/(\pi^{2}/16)$: A/P-BD is preferred under scarce reflect-power, SC/SC-BD when amplifier count is the bottleneck, and FC/SC-BD when both performance and energy efficiency matter.

\subsection{Extension to MISO/SIMO Links}\label{subsec:miso}
Alg.~\ref{alg:hbdris} extends to multi-antenna links by replacing $\mathbf{h}_{T,s,g}$ in \eqref{eq:Asg} with the rank-1 surrogate of \cite[Cor.~1]{santamaria2023bdris}: for a multiple-input single-output (MISO) link with channel $\bH_{T,s}\in\C^{M_s\times N_T}$, $\bH_{T,s}=\bU_{T,s}\bLambda_{T,s}\bV_{T,s}\herm$ via SVD, the surrogate is the dominant left singular vector $\mathbf{u}_{T,s,1}$ restricted to the rows of group $g$; the Tx then applies maximum ratio transmission (MRT) on $\mathbf{h}_{\mathrm{eq}}\herm=\mathbf{h}_R\herm\bP\bH_T$ \cite[Rem.~4]{santamaria2023bdris}. The same applies to single-input multiple-output (SIMO) and to the $K$-user SISO multiple access channel (MAC) with $\bH_T=[\sqrt{P_1}\mathbf{h}_1,\ldots,\sqrt{P_K}\mathbf{h}_K]$ \cite[Sec.~III-B]{santamaria2023bdris}.

\section{Numerical Results}\label{sec:results}
We numerically evaluate the proposed hybrid BD-RIS family in a SISO link with blocked direct path. The Rx is at $(50,0,2)$~m and the BD-RIS at $(40,2,5)$~m; the Tx is uniformly placed in a circle of radius $10$~m centered at the origin at $z=2$~m. The path-loss exponent is $\alpha=3.75$ with $\mathrm{PL}=-30-10\alpha\log_{10}d$ (dB), and both Tx-RIS and RIS-Rx links are i.i.d.\ Rayleigh. The noise PSD is $-174$~dBm/Hz, the bandwidth $B=2$~MHz, $P_t=20$~dBm, $P_{r,s}^{\max}=10$~dBm, $\delta_s^{2}=-100$~dBm, and $\beta_{\max}=\sqrt{10}$ \cite{long2021,ntougias2025hybrid}. Curves are averaged over $10^{3}$ Monte Carlo runs. Unless stated otherwise, $a=M_1/M=0.5$ and the BD-RIS group size is $M_G=4$.

\begin{figure}[t]
    \centering
    \includegraphics[width=0.95\columnwidth]{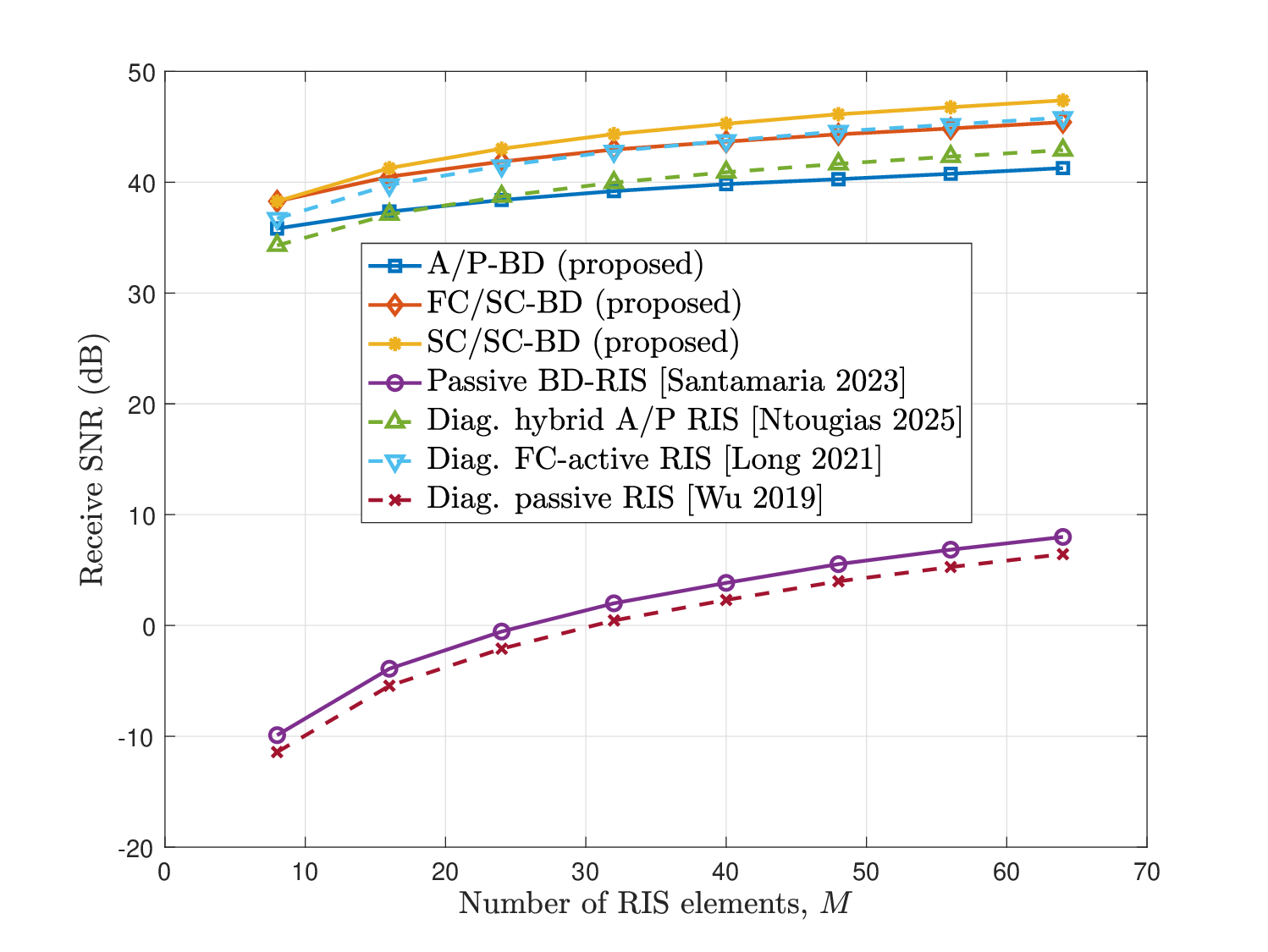}
    \caption{Receive SNR vs.\ $M$ for the proposed hybrid BD-RIS family and the relevant benchmarks ($a=0.5$, $M_G=4$).}
    \label{fig:abs_snr}
\end{figure}

\subsection{Absolute SNR vs.\ Architecture}\label{subsec:res_abs}
Fig.~\ref{fig:abs_snr} reports the absolute receive SNR versus $M$ for the family of hybrid BD-RIS and four benchmarks: passive BD-RIS \cite{santamaria2023bdris}, diagonal hybrid A/P RIS \cite{ntougias2025hybrid}, diagonal FC-active RIS \cite{long2021}, and diagonal passive RIS \cite{wu2019}. Two regimes emerge. The passive architectures lie in the $-10$ to $+8$~dB range, and the passive BD-RIS of \cite{santamaria2023bdris} is about $1.5$~dB above the diagonal passive RIS, in line with the asymptotic value of $10\log_{10}(16/\pi^2)\!\approx\!2.1$~dB \cite[Eq.~(62)]{li2023}. The active architectures lie in the $35$--$48$~dB range. The SC/SC-BD attains the highest SNR (about $47$~dB at $M{=}64$), $1.5$~dB above the diagonal FC-active RIS \cite{long2021} despite using only \emph{two} reflect-type amplifiers vs.\ $M$. The FC/SC-BD design tracks the diagonal FC-active RIS within $0.5$~dB. Interestingly, the A/P-BD curve sits about $1.5$~dB \emph{below} the diagonal hybrid A/P RIS of \cite{ntougias2025hybrid}; this is consistent with Fig.~\ref{fig:snr_gain} below, where we show that for small group sizes the BD-RIS spatial coupling on the passive RS does not compensate for the per-element amplification loss in the active RS unless $M_G$ is very large.

\begin{figure}[t]
    \centering
    \includegraphics[width=0.95\columnwidth]{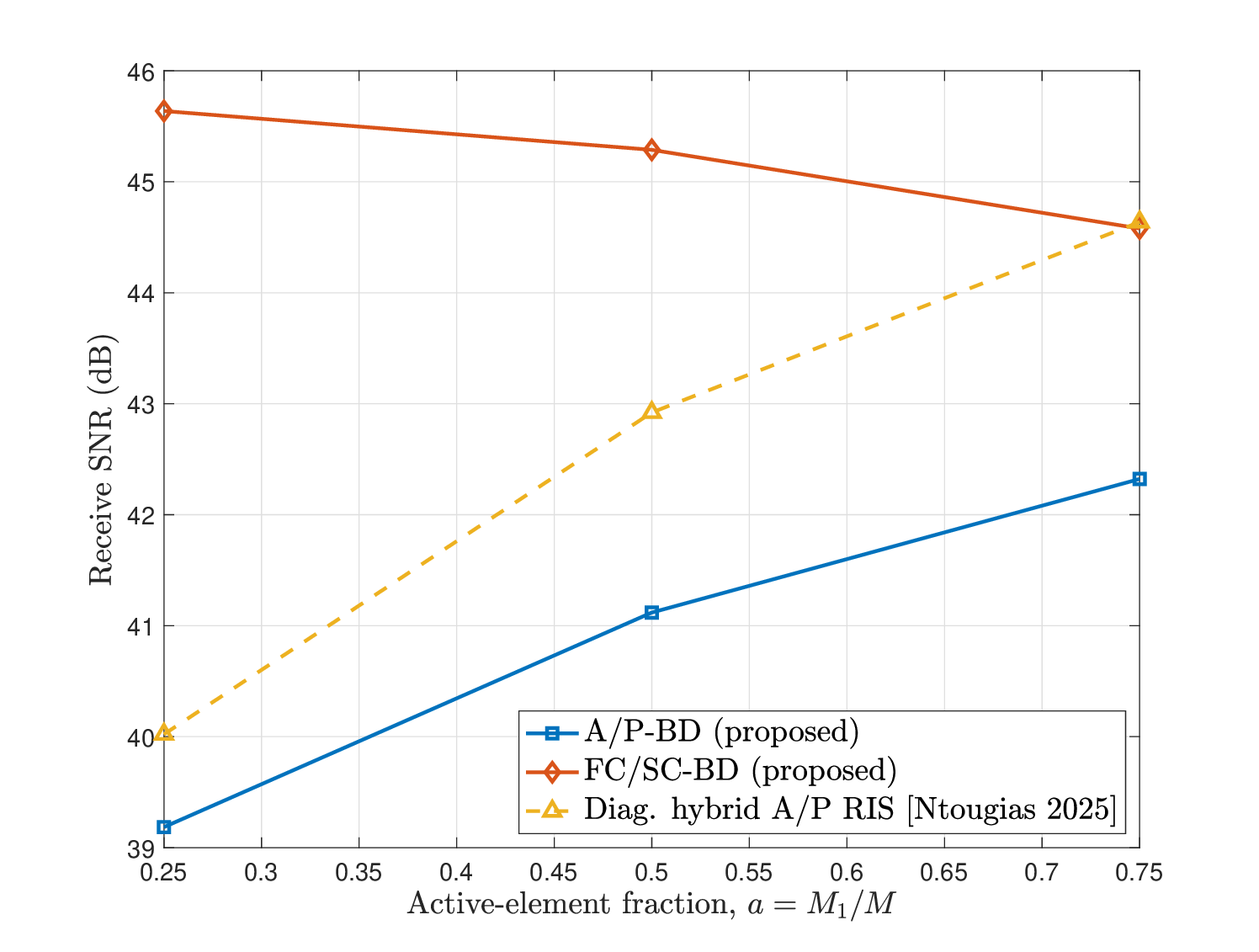}
    \caption{Receive SNR vs.\ $a=M_1/M$ ($M=64$, $M_G=4$).}
    \label{fig:snr_vs_a}
\end{figure}

\subsection{Impact of the Active-Element Fraction $a$}\label{subsec:res_a}
Fig.~\ref{fig:snr_vs_a} sweeps $a\in\{0.25,0.5,0.75\}$ for $M{=}64$ and $M_G{=}4$. The FC/SC-BD design is essentially flat at $\approx 45$~dB across the whole range and even slightly decreases with $a$, since enlarging RS$_1$ at the expense of RS$_2$ reduces the contribution of the SC-active back-end without proportionally increasing the FC-active front-end gain. In contrast, both the A/P-BD and the diagonal hybrid A/P RIS \cite{ntougias2025hybrid} grow monotonically with $a$ (from $\approx 39$~dB to $\approx 42.3$~dB and from $\approx 40$~dB to $\approx 44.6$~dB, respectively), as predicted by Proposition~\ref{prop:asy_AP}: enlarging the active RS strengthens the path-loss-compensated branch. Fig.~\ref{fig:snr_vs_a} therefore identifies FC/SC-BD as the most robust member of the family with respect to $a$, while A/P-BD is preferable when most of the surface can be made active.

\begin{figure}[t]
    \centering
    \includegraphics[width=0.95\columnwidth]{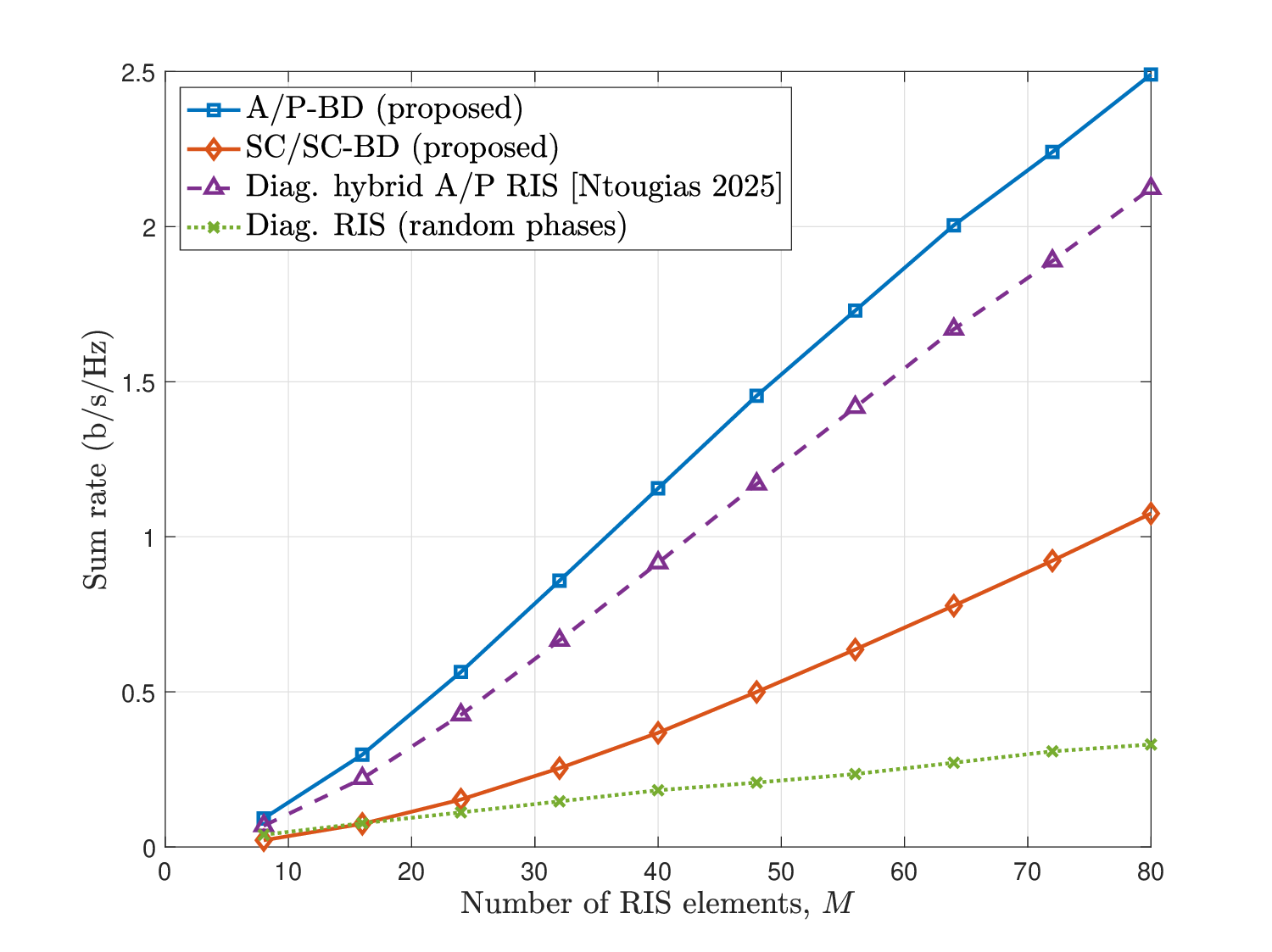}
    \caption{Sum rate vs.\ $M$ in a 2-user SISO MAC ($P_k=20$~dBm, $a=0.5$, $M_G=4$).}
    \label{fig:sumrate}
\end{figure}

\subsection{Sum-Rate in the 2-User SISO MAC}\label{subsec:res_sr}
Fig.~\ref{fig:sumrate} reports the sum rate versus $M$ for a 2-user SISO MAC with per-user power 20~dBm. The A/P-BD attains the highest sum rate over the whole range and exceeds the diagonal hybrid A/P RIS \cite{ntougias2025hybrid} by approximately $20\%$ at $M=80$ ($2.49$ vs.\ $2.13$~b/s/Hz). The reason is that, in the multi-user regime, the rank-1 surrogate of \cite[Cor.~1]{santamaria2023bdris} aligns the BD-RIS spatial coupling on the passive RS coherently with the dominant left singular vector of the user-to-RIS matrix, which is no longer the limitation observed in Fig.~\ref{fig:abs_snr}. The SC/SC-BD achieves about half the sum rate of A/P-BD, since its single shared amplifier per RS cannot exploit the per-user channel diversity, and diagonal RIS with random phases falls below all designs.

\begin{figure}[t]
    \centering
    \includegraphics[width=0.95\columnwidth]{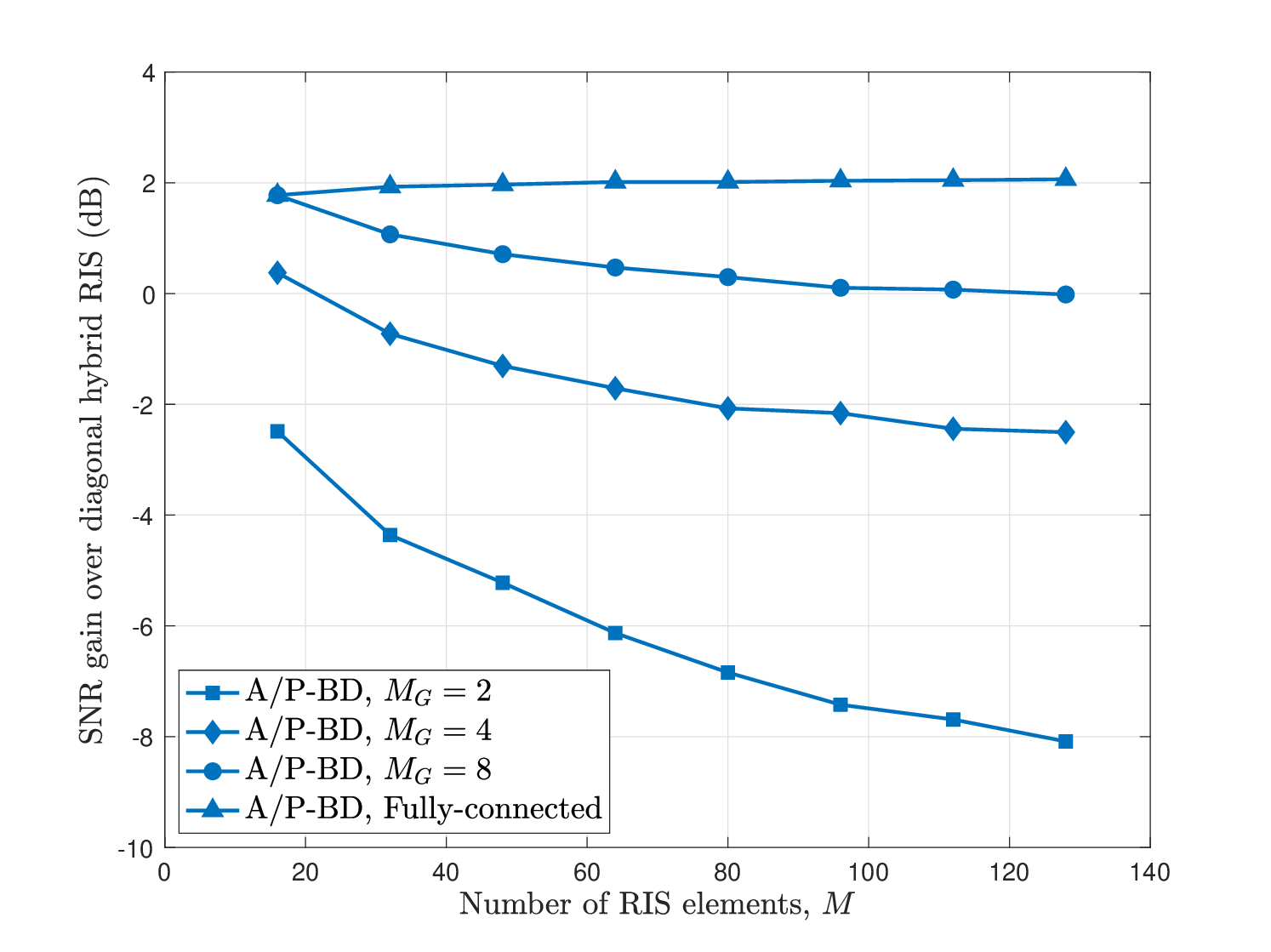}
    \caption{SNR gain of (F2) A/P-BD over the diagonal hybrid A/P RIS of \cite{ntougias2025hybrid} for different group sizes $M_G$ ($a=0.5$).}
    \label{fig:snr_gain}
\end{figure}

\subsection{Role of the Group Size $M_G$ in A/P-BD}\label{subsec:res_gain}
Fig.~\ref{fig:snr_gain} reports the SNR gain (in dB) of the A/P-BD over its diagonal counterpart \cite{ntougias2025hybrid} for $M_G\in\{2,4,8\}$ and the fully-connected case $M_G=M_s$. The fully-connected configuration provides a positive gain of about $2$~dB across the whole range, in line with the BD-RIS asymptotic bound \cite[Eq.~(62)]{li2023}. For finite group sizes, however, the A/P-BD becomes \emph{worse} than its diagonal counterpart as $M$ grows, with a loss of up to $-8$~dB at $M=128$ for $M_G=2$. The reason is structural: in the FC-active RS the $1/\sqrt{M_{G,1}}$ factor introduced by equal-power allocation across the $M_{G,1}$ patches per group (cf.\ \eqref{eq:Phi_fc}) is \emph{not compensated} by the unitary-symmetric BD-RIS coupling because the latter only redistributes phase, not amplitude. As $M$ grows with fixed $M_G$, the number of groups (and therefore the total $1/\sqrt{M_G}$ amplitude penalty across the surface) accumulates. This effect is absent in the fully-connected case ($M_G=M_s$), where there is a single group and the unitary BD-RIS coupling fully exploits the entire RS aperture. From an architectural perspective, this result indicates that the A/P-BD design is competitive with the diagonal hybrid A/P RIS only when the FC-active RS is implemented with a large group size or, ideally, fully-connected.

\subsection{Rate vs.\ Reflect-Power and Transmit Power Budgets}\label{subsec:res_budgets}
\begin{figure}[!t]
    \centering
    \includegraphics[width=0.95\columnwidth]{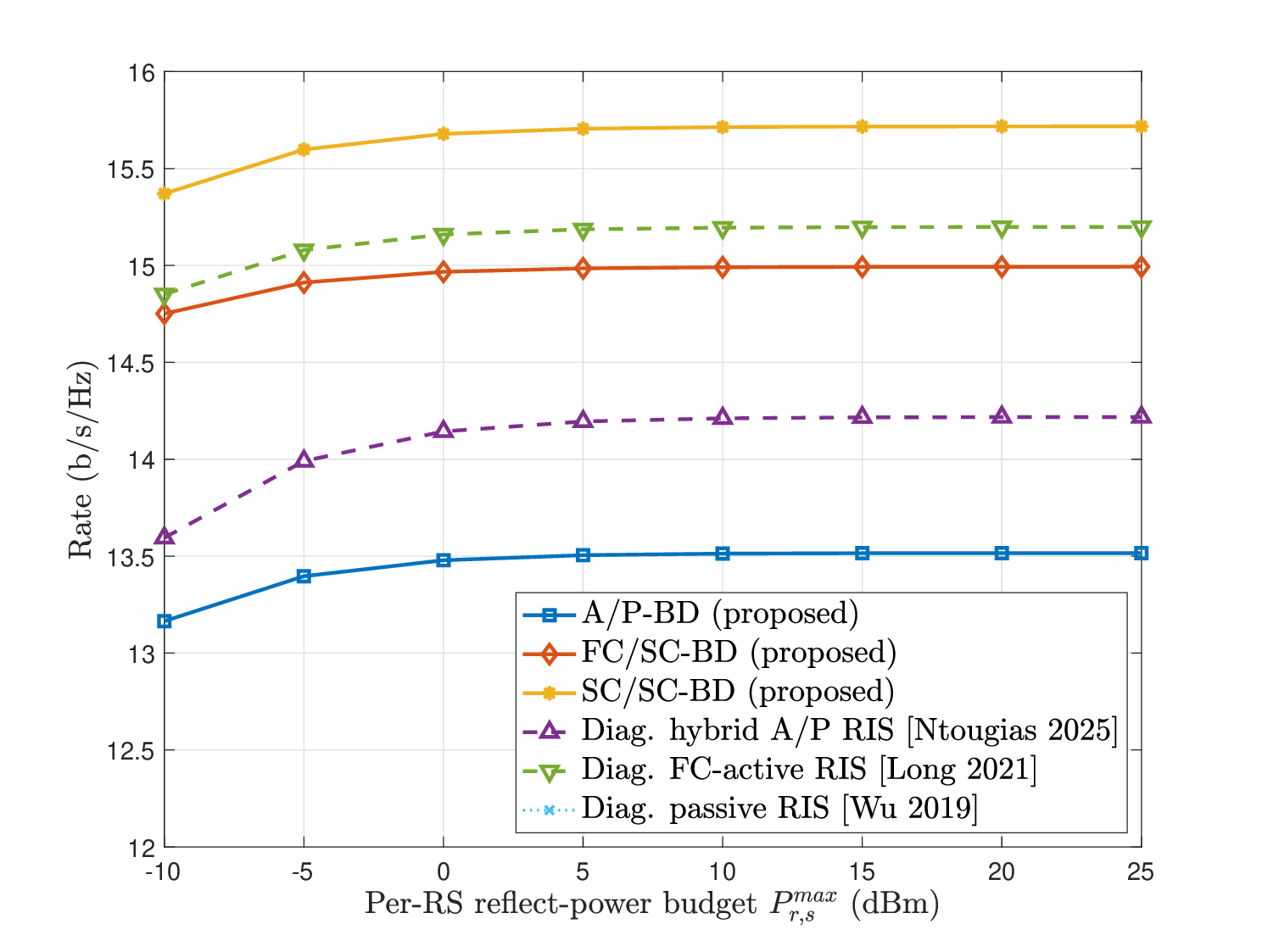}
    \caption{Rate vs.\ per-RS reflect-power budget $P_{r,s}^{\max}$ ($M=64$, $a=0.5$, $M_G=4$, $P_t=20$~dBm).}
    \label{fig:rate_vs_pr}
\end{figure}
\begin{figure}[!t]
    \centering
    \includegraphics[width=0.95\columnwidth]{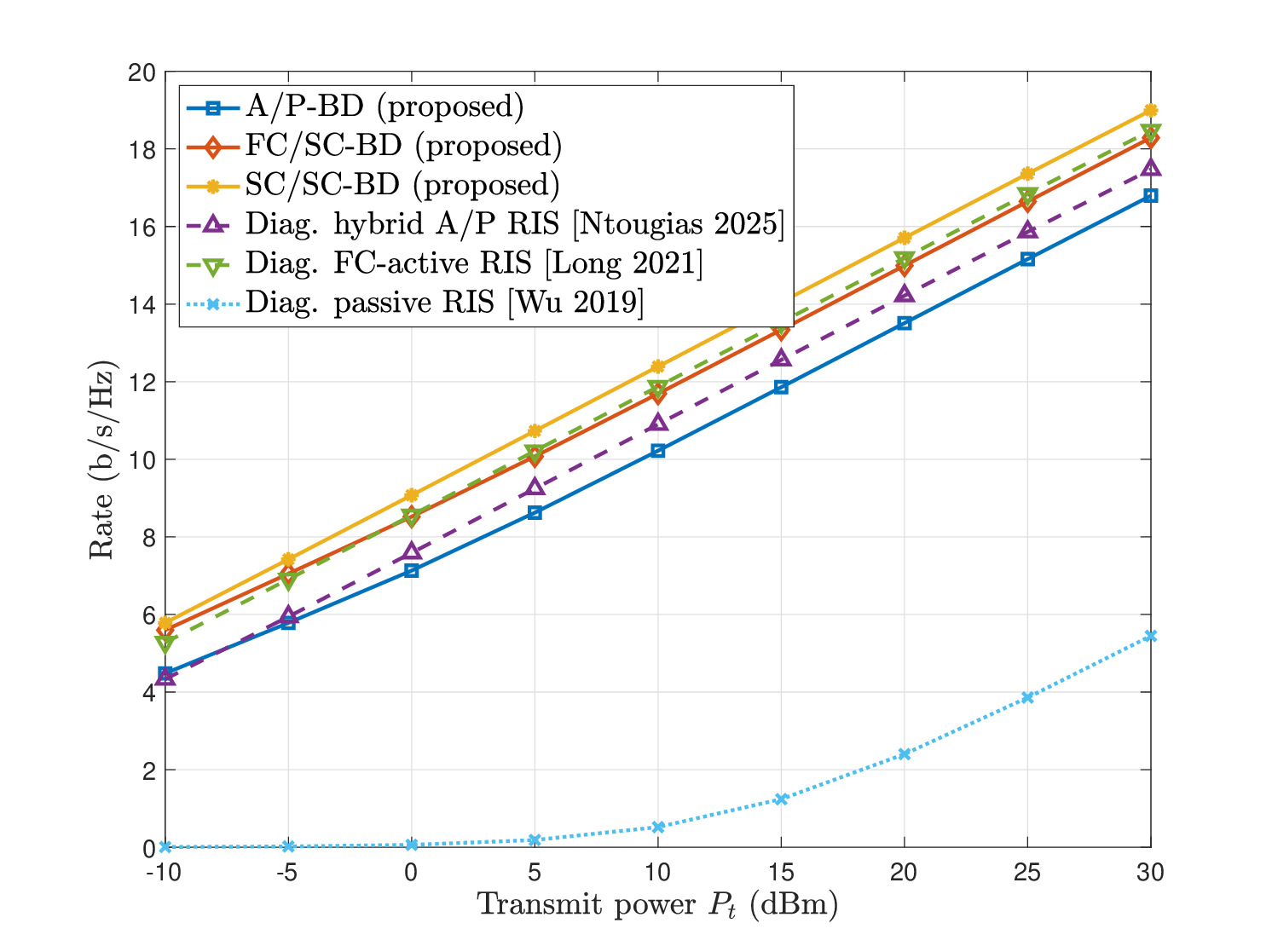}
    \caption{Rate vs.\ transmit power $P_t$ ($M=64$, $a=0.5$, $M_G=4$, $P_{r,s}^{\max}=10$~dBm).}
    \label{fig:rate_vs_pt}
\end{figure}
Fig.~\ref{fig:rate_vs_pr} sweeps the per-RS reflect-power budget $P_{r,s}^{\max}\in[-10,25]$~dBm at $P_t=20$~dBm and depicts the rate $R=\log_2(1+\gamma)$. All active architectures' rate saturates beyond $\approx 5$~dBm, since the transmit power becomes the bottleneck; SC/SC-BD attains the highest rate ($\approx 15.7$~b/s/Hz), $0.5$~b/s/Hz above the diagonal FC-active RIS \cite{long2021} despite using only two amplifiers, and FC/SC-BD tracks the latter within $0.2$~b/s/Hz. Fig.~\ref{fig:rate_vs_pt} sweeps $P_t\in[-10,30]$~dBm at $P_{r,s}^{\max}=10$~dBm, where the rate of all active designs grow linearly with $P_t$ (in dB) and SC/SC-BD again leads. Both figures confirm that, in agreement with Remark~\ref{rem:scaling}, amplifier sharing within an RS is not penalized in rate: the SC/SC-BD architecture, which uses the smallest amplifier count, achieves the highest rate over the entire range of both power budgets.

\section{Conclusion}\label{sec:conclusion}
We introduced a family of hybrid BD-RIS architectures. The SISO Max-SNR problem with blocked direct path was solved in closed form via per-group Takagi factorization and optimal per-cluster amplification. Results show that SC/SC-BD attains the highest SNR with only two amplifiers, exceeding the FC-active diagonal RIS of \cite{long2021}; FC/SC-BD is the most $a$-robust; and A/P-BD outperforms the diagonal A/P RIS of \cite{ntougias2025hybrid} in the 2-user MAC. Finally, for small $M_G$ the equal-power-allocation factor in the FC-active RS erodes the BD-RIS gain, making large or fully-connected groups preferable. Future work includes MIMO and multi-user extensions, elements allocation, and energy-efficient designs.

\bibliographystyle{IEEEtran}
\bibliography{refs}

\end{document}